\documentclass[12pt]{article}
\usepackage[utf8]{inputenc}
\usepackage{amsthm}
\usepackage{amsmath}
\usepackage{amsfonts}
\usepackage{epsfig}
\usepackage{colonequals}
\newtheorem{theorem}{Theorem}
\newtheorem{lemma}[theorem]{Lemma}
\newtheorem{corollary}[theorem]{Corollary}
\newtheorem{observation}[theorem]{Observation}
\newcommand{\wcol}{\text{wcol}}
\newcommand{\GG}{\mathcal{G}}
\newcommand{\FF}{\mathcal{F}}
\newcommand{\CC}{\mathcal{C}}
\newcommand{\KK}{\mathcal{K}}
\newcommand{\PP}{\mathcal{P}}

\title{A note on local search for hitting sets}

\author{Zdeněk Dvořák\thanks{Charles University, Prague, Czech Republic.
E-mail: {\tt rakdver@iuuk.mff.cuni.cz}.  Supported by the European Research Council (ERC) under the
European Union’s Horizon 2020 research and innovation programme (grant agreement No. 810115).}}
\date{}

\begin{document}
\maketitle

\begin{abstract}
Let $\pi$ be a property of pairs $(G,Z)$, where $G$ is a graph and $Z\subseteq V(G)$.
In the \emph{minimum $\pi$-hitting set problem}, given an input graph $G$,
we want to find a smallest set $X\subseteq V(G)$ such that $X$ intersects every set $Z\subseteq V(G)$ such that
$(G,Z)$ has the property $\pi$.  An important special case is that $\pi$ is satisfied by $(G,Z)$ exactly if $G[Z]$ is isomorphic to
one of graphs in a finite set $\FF$; in this \emph{minimum $\FF$-hitting set} problem, $X$ needs to hit all appearances of the graphs from
$\FF$ as induced subgraphs of $G$.  In this note, we show that the local search argument of Har-Peled and Quanrud
gives a PTAS for the minimum $\mathcal{F}$-hitting set problem for graphs from any class with polynomial expansion.
Moreover, we argue that the local search argument applies more generally to all properties $\pi$
such that one can test whether $X$ is a $\pi$-hitting set in polynomial time and $G[Z]$ has bounded diameter whenever $(G,Z)$
satisfies $\pi$; this is a common generalization of the minimum $\mathcal{F}$-hitting set problem and minimum $r$-dominating
set problem.  Finaly, we note that the analogous claim also holds for the dual problem of finding the maximum number of
disjoint sets $Z$ such that $(G,Z)$ has the property $\pi$; this generalizes maximum $F$-matching, maximum induced $F$-matching,
and maximum $r$-independent set problems.
\end{abstract}

Let $H$ and $G$ be graphs.  If a graph isomorphic to $H$ is obtained from a subgraph of $G$ by contracting pairwise vertex-disjoint
subgraphs, each of radius at most $r$, we say that $H$ is an \emph{$r$-shallow minor} of $G$.  The \emph{density} of a graph $H$
is $|E(H)|/|V(H)|$.  For a function $f:\mathbb{Z}_0^+\to\mathbb{R}_0^+$, a graph $G$ has \emph{expansion bounded by $f$}
if for every non-negative integer $r$, every $r$-shallow minor of $G$ has density at most $f(r)$.
A class of graphs $\GG$ has expansion bounded by $f$ if all graphs from $\GG$ do.
We say $\GG$ has \emph{bounded expansion} if $\GG$ has expansion bounded by some function $f$, and \emph{polynomial expansion}
if $\GG$ has expansion bounded by a polynomial.

Let $\pi$ be a property of pairs $(G,Z)$, where $G$ is a graph and $Z\subseteq V(G)$ is non-empty.
The \emph{diameter} of $\pi$ is the maximum diameter of $G[Z]$ over all pairs $(G,Z)$ satisfying the property $\pi$,
or $\infty$ if there exist such pairs with the diameter of $G[Z]$ arbitrarily large (this includes the case that $G[Z]$ is disconnected
for some $(G,Z)$ satisfying $\pi$).  For a graph $G$, a set $X\subseteq V(G)$ is a \emph{$\pi$-hitting set} if
$X$ intersects every $Z\subseteq V(G)$ such that $(G,Z)$ satisfies $\pi$.  Let $\gamma_\pi(G)$ denote the minimum size of a $\pi$-hitting set in $G$.
For a graph class $\GG$, we say that \emph{$\pi$-hitting set is efficiently testable in $\GG$} if there
exists a polynomial-time algorithm that for a graph $G\in\GG$ and a set $X\subseteq V(G)$ decides whether 
$X$ is a $\pi$-hitting set. 
Let us mention two important examples:
\begin{itemize}
\item For a finite set $\FF$ of graphs, if $\pi$ is the property that $G[Z]$ is isomorphic to a graph in $\FF$,
we obtain the \emph{$\FF$-hitting set} problem of finding a smallest set $X\subseteq V(G)$ such that $G-X$ is \emph{$\FF$-free},
i.e., does not contain any induced subgraph isomorphic to an element of $\FF$. 
The diameter of $\pi$ is equal to the maximum of diameters of the graphs in $\FF$.
\item Let $r$ be a positive integer, and let $\pi$ be satisfied by $(G,Z)$ iff there exists $z\in V(G)$ such that $Z$ consists
exactly of vertices at distance at most $r$ from $z$ in $G$.  Then $X$ is a $\pi$-hitting set iff it is an \emph{$r$-dominating set},
i.e., every vertex of $G$ is at distance at most $r$ from $G$.  The diameter of $\pi$ is at most $2r$.
\end{itemize}
In both cases, $\pi$-hitting set is efficiently testable in all graphs.

For any $\pi$ such that $\pi$-hitting set is efficiently testable and a positive integer $c$,
one can use the following \emph{$c$-local search} heuristic to find a (hopefully small) $\pi$-hitting set:
\begin{enumerate}
\item Let $X=V(G)$.
\item Iterate over all sets $Y\subseteq V(G)$ of size at most $c$ such that $|X\triangle Y|<|X|$:
\begin{itemize}
\item If $X \triangle Y$ is a $\pi$-hitting set, then let $X\colonequals X \triangle Y$ and repeat the step 2.
\end{itemize}
\end{enumerate}
Of course, there are in general no guarantees that the result of local search is close to optimal.  However, Har-Peled and Quanrud~\cite{har2015approximationjrnl}
proved the following claim for minimum vertex cover ($\{K_2\}$-hitting set) and for minimum $r$-dominating set:
For every class of graphs $\GG$ with polynomial expansion and every $\varepsilon>0$, there exists $c_\varepsilon$
such that $c_\varepsilon$-local search results in a $\pi$-hitting set of size at most $(1+\varepsilon)\gamma_\pi(G)$
for every graph $G\in\GG$.  That is, local search gives a PTAS when resticted to any graph class with polynomial expansion.
The main result of this note is the following common generalization of these results.
\begin{theorem}\label{thm-main}
Let $\GG$ be a class of graphs with polynomial expansion.  For every property $\pi$ of finite diameter and every $\varepsilon>0$,
there exists $c$ such that $c$-local search results in a $\pi$-hitting set of size at most $(1+\varepsilon)\gamma_\pi(G)$.
\end{theorem}
It is also natural to consider the dual notion.  A set $\PP$ of pairwise disjoint subsets of vertices of $G$ is a \emph{$\pi$-packing}
if $(G,P)$ has the property $\pi$ for every $P\in\PP$; it is an \emph{induced $\pi$-packing} if additionally no edge of $G$
has ends in distinct elements of $\PP$.  For example:
\begin{itemize}
\item For a graph $F$, if $\pi$ is the property that $G[P]$ is isomorphic to $F$,
we obtain the \emph{$F$-matching} (and \emph{induced $F$-matching}) problem
of finding the maximum number of disjoint (and non-adjacent) copies of $F$ in the input graph.
\item Let $r$ be a positive integer, and let $\pi$ be satisfied by $(G,Z)$ iff there exists $z\in V(G)$ such that $Z$ consists
exactly of vertices at distance at most $r$ from $z$ in $G$.  Then a $\pi$-packing corresponds to a $2r$-independent set (a set
of vertices is \emph{$t$-independent} if the distance between any distinct elements is greater than $t$), and an induced $\pi$-packing
corresponds to a $(2r+1)$-independent set.
\end{itemize}
Let $\alpha_\pi(G)$ denote the maximum size of a $\pi$-packing and $\alpha'_\pi(G)$
the maximum size of an induced $\pi$-packing in $G$.
For a graph class $\GG$, we say that \emph{(induced) $\pi$-packing extension is tractable in $\GG$} if for every positive integer
$c$, there exists a polynomial-time algorithm that for a graph $G\in\GG$ and a set $T\subseteq V(G)$ finds
an (induced) $\pi$-packing $\PP_0$ of size at least $c$ and such $T\cap\bigcup \PP_0=\emptyset$, or decides no such $\PP_0$ exists.
If $\pi$ has this property, we can run the following maximization variant of the $c$-local search algorithm:
\begin{enumerate}
\item Let $\PP=\emptyset$.
\item Iterate over all $Y\subseteq \PP$ of size less than $c$:
\begin{itemize}
\item If $\PP\setminus Y$ can be extended to an (induced) $\pi$-packing $\PP'$ of size greater than $|\PP|$,
then let $\PP\colonequals \PP'$ and repeat the step 2.
\end{itemize}
\end{enumerate}
Step 2 is implemented by testing whether there exists an (induced) $\pi$-packing $\PP_0$ of size $|Y|+1$ disjoint from $\bigcup_{P\in \PP\setminus Y} P$
(or this set and its neighbors in the induced case).
As our second result, we note that this algorithm again gives a PTAS in any class of graphs with polynomial
expansion; the argument is less involved and essentially implicitly appears in~\cite{har2015approximationjrnl}.
\begin{theorem}\label{thm-mainpack}
Let $\GG$ be a class of graphs with polynomial expansion.  For every property $\pi$ of finite diameter and every $\varepsilon>0$,
there exists $c$ such that the (induced) maximization $c$-local search results in an (induced) $\pi$-packing of size at least
$(1-\varepsilon)\alpha_\pi(G)$ in the non-induced case and at least $(1-\varepsilon)\alpha'_\pi(G)$ in the induced case.
\end{theorem}

The time complexity of the algorithms from Theorems~\ref{thm-main} and \ref{thm-mainpack} in a straightforward
implementation is $O(n^{c+1})$ times the complexity of testing or extension, and thus the exponent of the time
complexity depends on the desired precision; i.e., we obtain a PTAS, not an FPTAS.  However, in all the examples
that we considered, the property $\pi$ is \emph{FO-definable} in the following sense: There exists a formula $\varphi$
in the first-order logic (allowing quantification over vertices, but not over sets of vertices or over edges)
using the adjacency predicate, with free variables $x_1$, \ldots, $x_m$ and $z$ such that $(G,Z)$ with $Z$ non-empty has the property
$\pi$ if and only if for some $v_1,\ldots,v_m\in G$,
$$Z=\{v\in V(G):G\models \varphi(v_1,\ldots,v_m,v)\}.$$
E.g., for the property $\pi$ used to define $r$-dominating and $2r$-independent set problems, we can set
$\varphi(x_1,z)\equiv\text{dist}_r(x_1,z)$,
where
$$\text{dist}_r(x,y)\equiv (\exists p_0,\ldots,p_r)\;x=p_0\land y=p_r\land \bigwedge_{i=1}^r (p_{i-1}=p_i\lor p_{i-1}p_i\in E).$$
In this case, we can express the testing and extension in first-order logic.  Let us demonstrate this on the most difficult
example, the step 2 of the $c$-local search for maximum $\pi$-packing.  Let $\PP$ be the current $\pi$-packing.  For each element $P$ of $\PP$,
color the edges of a BFS spanning tree of $G[P]$ red; we will access this information about edge colors using a new binary predicate $R$.
Let $\text{dist}^R_r(x,y)$ be the predicate defined analogously to $\text{dist}_r(x,y)$, but with $E$ replaced by $R$.
Let $T$ be the unary predicate interpreted as $\bigcup \PP$.
We can delete at most $c-1$ elements from $\PP$ and add $c$ other elements to obtain a $\pi$-packing exactly if the following
formula is satisfied:
\begin{align*}
(\exists& r_1,\ldots,r_{c-1},x_1^1,\ldots,x_m^1,\ldots,x_1^c,\ldots,x_m^c)\\
&\bigwedge_{j=1}^c (\exists z)\;\varphi(x_1^j,\ldots,x_m^j,z)\\
&\land \bigwedge_{j=1}^c (\forall z)\;\Bigl(\varphi(x_1^j,\ldots,x_m^j,z)\land z\in T\Rightarrow \bigvee_{i=1}^{c-1} r_i\in T\land \text{dist}_{2r}^R(z,r_i)\Bigr)\\
&\land \bigwedge_{1\le j_1 < j_2\le c} (\forall z_1,z_2)\;\bigl(\varphi(x_1^{j_1},\ldots,x_m^{j_1},z_1)\land \varphi(x_1^{j_2},\ldots,x_m^{j_2},z_2)\Rightarrow z_1\neq z_2\bigr).
\end{align*}
The elements to be removed from $\PP$ are those intersected by $r_1$, \ldots, $r_{c-1}$, while those added are given by the formula $\varphi$
with parameters $x_1^j$, \ldots, $x_m^j$ for $j\in\{1,\ldots,c\}$.
\begin{itemize}
\item The second line expresses that for each $j$, the formula $\varphi$ parameterized by the $m$-tuple $x_1^j$, \ldots, $x_m^j$ defines a non-empty set $Z_j$.
\item The third line states that $Z_j$ intersects $T$ only in the elements containing $r_1$, \ldots, $r_{c-1}$
that are to be removed from $\PP$.
\item The final line describes that the sets $Z_1$, \ldots, $Z_m$ are disjoint.
\end{itemize}

Dvořák, Král' and Thomas~\cite{dvorak2013testing} described a data structure that for each graph from a fixed class with bounded expansion
after a linear-time preprocessing allows one in constant time perform the following operations:
\begin{itemize}
\item Recolor an edge.
\item Change the membership of a vertex in a set interpreting an unary predicate.
\item Test whether a fixed first-order sentence $\psi$ is satisfied, and
\item if so, find a satisfying assignment for the initial segment of existential quantifiers of $\psi$.
\end{itemize}
With this data structure, we conclude that for a FO-definable property of bounded diameter,
$n$-vertex graph $G$ from any fixed class with bounded expansion, and any fixed $c$, we can run
\begin{itemize}
\item $c$-local search for minimum $\pi$-hitting in time $O(n)$, and
\item (induced) $c$-local search for maximum $\pi$-packing in time $O(n^2)$.
\end{itemize}
The reason for the quadratic time in the second case is due to the need to recolor the red edges after each iteration,
which may result in $O(n)$ time in case the sets of the (induced) $\pi$-packing have unbounded size.
In case they have bounded size, e.g., for the maximum (induced) $F$-matching problem, the time complexity becomes $O(n)$.
Moreover, even if the elements of the packing have unbounded size, it may be possible to use a different way how to
describe the packing in the first-order logic (e.g., in the case of maximum $t$-independent set, we can just mark the vertices
forming the set), again resulting in total time $O(n)$.

\section{Tools}

To prove Theorem~\ref{thm-main}, we use several standard tools from the theory of classes with bounded expansion.

\subsection*{Generalized coloring numbers}

Let $\prec$ be a linear ordering of vertices of a graph $G$.  For a non-negative integer $r$, a vertex $u$ is \emph{weakly $(r,\prec)$-reachable}
from a vertex $v\succeq u$ if there exists a path $P$ from $v$ to $u$ in $G$ of length at most $r$ such that all vertices in
$V(P)\setminus\{u\}$ are greater than $u$ in the ordering $\prec$.  Let $L_{r,\prec}(v)$ denote the set of all vertices that
are weakly $(r,\prec)$-reachable from $v$, and let $R_{r,\prec}(u)$ denote set of all vertices $v$ from which $u$ is weakly $(r,\prec)$-reachable.
The \emph{weak $r$-coloring number} of $(G,\prec)$ is the maximum of $|L_{r,\prec}(v)|$ over
all vertices $v\in V(G)$.  The \emph{weak $r$-coloring number} $\wcol_r(G)$ of $G$ is the minimum weak $r$-coloring number of $(G,\prec)$ over all
linear orderings $\prec$ of $V(G)$.

\begin{theorem}[Zhu~\cite{zhu2009colouring}]\label{thm-zhu}
For every class $\GG$ with bounded expansion and every $r\ge 0$, there exists an integer $b$ such that $\wcol_r(G)\le b$ for every $G\in\GG$.
\end{theorem}

\subsection*{Expansion of shallow packings}

A collection $\CC$ of subsets of vertices of a graph $G$ is \emph{$(\omega, t)$-shallow} if every vertex of $G$ appears in
at most $\omega$ elements of $\CC$ and if $G[C]$ has radius at most $t$ for every $C\in \CC$.  The \emph{packing graph} $G[\CC]$
is defined as the graph with vertex set $\CC$ where $C_1$ and $C_2$ are adjacent if there exist $v_1\in C_1$ and $v_2\in C_2$
such that $v_1=v_2$ or $v_1v_2\in E(G)$.  In particular, an $r$-shallow minor is the packing graph of a $(1,r)$-shallow collection.
\begin{theorem}[Har-Peled and Quanrud~\cite{har2015approximationjrnl}]
Let $G$ be a graph with expansion bounded by a function $f$.  For every $(\omega, t)$-shallow collection $\CC$ of subsets of vertices of $G$,
the graph $G[C]$ has expansion bounded by the function
$$f'(r)=5\omega^2(2t+1)^2(2r + 1)^2f((2t+1)r+t).$$
\end{theorem}

\begin{corollary}\label{cor-preserve}
For every class $\GG$ with polynomial expansion and all non-negative integers $\omega$ and $t$, there exists
a class $\GG'$ with polynomial expansion such that the packing graph of any $(\omega, t)$-shallow collection of subsets of vertices
of a graph $G\in \GG$ belongs to $\GG'$.
\end{corollary}

\subsection*{Polynomial expansion and strongly sublinear separators}

A \emph{balanced separator} in a graph $G$ is a set $S\subseteq V(G)$ such that each component of $G-S$ has at most $\tfrac{2}{3}|V(G)|$
vertices.  For a function $s:\mathbb{Z}^+\to\mathbb{R}_0^+$, a class of graphs $\GG$ has \emph{$s$-separators} if for every $G\in \GG$,
every subgraph $H$ of $G$ has a balanced separator of size at most $s(|V(H)|)$.  A class $\GG$ has \emph{strongly sublinear separators}
if there exists $\varepsilon>0$ and a function $s(n)=O(n^{1-\varepsilon})$ such that $\GG$ has $s$-separators.

\begin{theorem}[Dvořák and Norin~\cite{dvorak2016strongly}]
A class of graphs has strongly sublinear separators if and only if it has polynomial expansion.
\end{theorem}

Let $\KK$ be a system of subsets of vertices of a graph $G$.  We say that $\KK$ is a \emph{cover} of $G$ if $\bigcup_{K\in\KK} G[K]=G$.
For $K\in \KK$, let $\partial K$ be the set of vertices of $K$ that belong to more than one element of $\KK$.  The \emph{excess}
of a cover $\KK$ is $\sum_{K\in \KK} |\partial K|$.  The following lemma is folklore (and the proof can be found e.g. in~\cite{har2016notes}).

\begin{lemma}\label{lemma-break}
Let $\GG$ be a class of graphs with strongly sublinear separators.  There exists a polynomial $p$ such that for
every $\varepsilon>0$, every graph $G\in \GG$ has a cover with excess at most $\varepsilon|V(G)|$ consisting of sets of size
at most $p(1/\varepsilon)$.
\end{lemma}

\section{The proof}

The proof of Theorem~\ref{thm-main} follows the idea of~\cite{har2015approximationjrnl}, with the following crucial difference: Instead of
comparing the actual solution with the optimal one, we compare it with a suitably chosen enlargement of the optimal solution.
For a set $O$ of vertices of a graph $G$, a linear ordering $\prec$ of $V(G)$, and a non-negative integers $r$ and $m$,
we say a vertex $v$ is \emph{$(r,\prec,m,O)$-rich} if $|R_{r,\prec}(v)\cap O|\ge m$.

\begin{observation}\label{obs-fewrich}
Let $G$ be a graph, $O$ a set of vertices of $G$, $r$ and $m$ non-negative integers, and $\prec$
a linear ordering of vertices of $G$.  Let $O'$ be the set of all $(r,\prec,m,O)$-rich vertices.
If the weak $r$-coloring number of $(G,\prec)$ is at most $b$, then
$$|O'|\le \frac{b}{m}|O|.$$
\end{observation}
\begin{proof}
By the definition of richness, there are at least $m|O'|$ pairs of vertices $(v,u)$ such that $v\in O'$
and $u\in R_{r,\prec}(v)\cap O$.  However, if $u\in R_{r,\prec}(v)$, then $v\in L_{r,\prec}(u)$,
and $|L_{r,\prec}(u)|\le b$ by the assumptions.  Hence, the number of such pairs is also at most $b|O|$.
\end{proof}

Let us consider the setting of Observation~\ref{obs-fewrich}, and let $A$ be another set of vertices of $G$.
Let $\CC_{r,\prec,m,O,A}=\{C_v:v\in A\cup O\cup O'\}$ be a system of subsets of vertices of $G$ defined
as follows:
\begin{itemize}
\item For $v\in O'$, let $C_v=R_{r,\prec}(v)$.
\item For $v\in O\setminus O'$, let
$$C_v=\bigcup_{x\in L_{r,\prec}(v)\setminus O'} R_{r,\prec}(x).$$
\item For $v\in A\setminus (O\cup O')$, let $C_v=\{v\}$.
\end{itemize}

\begin{lemma}\label{lemma-main}
Let $G$ be a graph, $A$ and $O$ sets of vertices of $G$, $r$ and $m$ non-negative integers, and $\prec$
a linear ordering of vertices of $G$.  Let $O'$ be the set of all $(r,\prec,m,O)$-rich vertices.
Let $H$ be the packing graph of $\CC=\CC_{r,\prec,m,O,A}$.
If $F$ is a subgraph of $G$ of diameter at most $r$ and $V(F)\cap O\neq\emptyset$, then
there exists $u\in V(F)\cap(O\cup O')$ such that $C_uC_v\in E(H)$ for every $v\in V(F)\cap A\setminus (O\cup O')$.
Moreover, if the weak $r$-coloring number of $(G,\prec)$ is at most $b$, then
$\CC$ is $(bm+1,2r)$-shallow.
\end{lemma}
\begin{proof}
Let $x$ be the smallest vertex of $V(F)$ in the ordering $\prec$.  
Since $F$ has diameter at most $r$, for every $y\in V(F)$, there exists a path from $y$ to $x$ of length at most
$r$ in $F$, and by the choice of $x$, the vertices of this path appear after $x$ in the ordering $\prec$.
Consequently $y\in R_{r,\prec}(x)$ and $x\in L_{r,\prec}(y)$.

If $x\in O'$, let $u=x$, otherwise choose $u$ as an arbitrary vertex of $V(F)\cap O$.
In the former case, we have $V(F)\subseteq R_{r,\prec}(u)=C_u$.
In the latter case, we have $x\in L_{r,\prec}(u)\setminus O'$, and
thus $V(F)\subseteq R_{r,\prec}(x)\subseteq C_u$.  Hence, $C_uC_v\in E(H)$ for every $v\in V(F)\cap A\setminus (O\cup O')$.

Next, let us bound the diameter of the sets in $\CC$.
For any $z\in V(G)$ and for every vertex $w$ of $R_{r,\prec}(z)$, there is a path of length at most $r$ from $w$ to $z$
with all vertices appearing after $z$ in $\prec$; observe that the vertices of this path also belong to $R_{r,\prec}(z)$.
Therefore, every vertex of $R_{r,\prec}(z)$ is at distance at most $r$ from $z$ in $G[R_{r,\prec}(z)]$,
and thus $G[R_{r,\prec}(z)]$ has diameter at most $2r$.  For $v\in O\cup O'$, we conclude that $C_v$
is the union of vertex sets of graphs of diameter at most $2r$, each of them containing $v$, and thus $G[C_v]$ has radius at most $2r$.
For $v\in A\setminus (O\cup O')$, $G[C_v]$ has radius $0$.

Finally, let us argue that no vertex belongs to too many of the sets in $\CC$. If a vertex $z$ belongs to $C_v$ for $v\in O'$,
then $z\in R_{r,\prec}(v)$, and thus $v$ one of at most $b$ vertices in $L_{r,\prec}(z)$.  If $z$ belongs to $C_v$ for
$v\in O\setminus O'$, then $z\in R_{r,\prec}(x)$ for some $x\in L_{r,\prec}(v)\setminus O'$; there are at most
$b$ choices for $x$ in $L_{r,\prec}(z)$, and for fixed $x\not\in O'$, there are less than $m$ vertices $v\in O$
such that $x\in L_{r,\prec}(v)$ (or equivalently, $v\in R_{r,\prec}(x)\cap O$).  Finally, vertex $z$ belongs to $C_v$ for
at most one $v\in A\setminus (O\cup O')$.  Therefore, $z$ belongs to $C_v$ for at most $b+b(m-1)+1$ vertices $v\in A\cup O\cup O'$.
Therefore, $\CC$ is $(bm+1,2r)$-shallow.
\end{proof}

We are now ready to prove the main result.
\begin{proof}[Proof of Theorem~\ref{thm-main}]
Without loss of generality, assume that $\varepsilon\le 1$.  Let $r$ be the diameter of $\pi$.
By Theorem~\ref{thm-zhu}, there exists $b$ such that $\wcol_r(G)\le b$ for every $G\in \GG$.
Let $m=\lceil 3b/\varepsilon\rceil$.  Let $\GG'$ be the class of graphs with polynomial expansion
given by Corollary~\ref{cor-preserve} with $\omega=bm+1$ and $t=2r$.  Let $p$ be the polynomial
from Lemma~\ref{lemma-break} for this class $\GG'$.
Let us define $c=\lceil p(18/\varepsilon)\rceil$.

Consider any graph $G\in\GG$.  Let $O$ be a $\pi$-hitting set in $G$ of size $\gamma_\pi(G)$, and let $A$ be the $\pi$-hitting
set returned by the $c$-local search algorithm.  Let $O'$ be the set of all $(r,\prec,m,O)$-rich vertices in $G$;
by Observation~\ref{obs-fewrich}, we have $|O'|\le \tfrac{b}{m}|O|\le \tfrac{\varepsilon}{3}|O|$.
Let $\CC=\CC_{r,\prec,m,O,A}$; by Lemma~\ref{lemma-main}, $\CC$ is $(bm+1,2r)$-shallow, and
by Corollary~\ref{cor-preserve}, the packing graph $H=G[\CC]$ belongs to $\GG'$.  For any $Y\subseteq V(H)$,
let $\overline{Y}=\{v\in V(G):C_v\in Y\}$.  By Lemma~\ref{lemma-break},
there exists a cover $\KK$ of $H$ with excess at most $\tfrac{\varepsilon}{18}|V(H)|\le \tfrac{\varepsilon}{18}(|A|+|O|+|O'|)\le \tfrac{\varepsilon}{6}|A|$
such that each element of $\KK$ has size at most $p(18/\varepsilon)\le c$.

Consider any $K\in \KK$, and let $A'=(A\setminus (\overline{K\setminus\partial K}))\cup ((O\cup O')\cap \overline{K})$.
We claim that $A'$ is a $\pi$-hitting set.  Indeed, suppose for a contradiction that there exists $Z\subseteq V(G)$ such
that $(G,Z)$ satisfies the property $\pi$, but $A'\cap Z=\emptyset$.  Since $A$ is a $\pi$-hitting set, there exists $v\in A\cap Z$,
and since $v\not\in A'$, we have $v\in (\overline{K\setminus\partial K})\setminus (O\cup O')$.
Since $O$ is a $\pi$-hitting set, we have $Z\cap O\neq\emptyset$,
and by Lemma~\ref{lemma-main} applied to $F=G[Z]$, there exists $u\in Z\cap (O\cup O')$ such that $C_uC_v\in E(H)$.
Since $\KK$ is a cover of $H$, there exists $K'\in \KK$ such that $C_uC_v\in E(H[K'])$.
Since $C_v\in K\setminus\partial K$, it follows that $K'=K$, and thus this implies $C_u\in K$.
Therefore, we have $u\in (O\cup O')\cap \overline{K}$, and $u\in A'\cap Z$, which is a contradiction.

Note that $A\triangle A'\subseteq \overline{K}$, and thus $|A\triangle A'|\le c$.  Hence, the $c$-local search algorithm
considered the solution $A'$ and did not improve $A$ to $A'$, which implies that $|A'|\ge |A|$.
We conclude that
$$|(O\cup O')\cap \overline{K}|\ge |A\cap (\overline{K\setminus\partial K})|$$
for every $K\in \KK$.  Therefore, denoting by $s$ the excess of $\KK$, we have
\begin{align*}
|A|&\le \sum_{K\in\KK} |A\cap \overline{K}|\le s+\sum_{K\in\KK} |A\cap (\overline{K\setminus\partial K})|\\
&\le s+\sum_{K\in\KK} |(O\cup O')\cap \overline{K}|\le 2s+\sum_{K\in\KK} |(O\cup O')\cap (\overline{K\setminus\partial K})|\\
&\le 2s+|O|+|O'|\le \tfrac{\varepsilon}{3}|A|+(1+\varepsilon/3)|O|.
\end{align*}
Therefore,
$$|A|\le \frac{1+\varepsilon/3}{1-\varepsilon/3}|O|\le (1+\varepsilon)|O|,$$
as required.
\end{proof}

The proof for the packing version is much simpler.
\begin{proof}[Proof of Theorem~\ref{thm-mainpack}]
Without loss of generality, assume that $\varepsilon\le 1$.  Let $r$ be the diameter of $\pi$.
Let $\GG'$ be the class of graphs with polynomial expansion given by Corollary~\ref{cor-preserve}
with $\omega=2$ and $t=r$.  Let $p$ be the polynomial from Lemma~\ref{lemma-break} for this class $\GG'$.
Let us define $c=\lceil p(4/\varepsilon)\rceil$.

Consider any graph $G\in\GG$.  Let $O$ be an (induced) $\pi$-packing in $G$ of size $\alpha_\pi(G)$ (or $\alpha'_\pi(G)$),
and let $A$ be the (induced) $\pi$-packing returned by the $c$-local search algorithm. 
Let $\CC=A\cup O$, and note that $\CC$ is $(2,r)$-shallow.  By Corollary~\ref{cor-preserve}, the packing graph $H=G[\CC]$
belongs to $\GG'$.  By Lemma~\ref{lemma-break}, there exists a cover $\KK$ of $H$ with excess at most
$\tfrac{\varepsilon}{4}|V(H)|\le \tfrac{\varepsilon}{4}(|A|+|O|)\le \tfrac{\varepsilon}{2}|O|$
such that each element of $\KK$ has size at most $p(4/\varepsilon)\le c$.

Consider any $K\in \KK$, and let $A'=(A\setminus K)\cup (O\cap (K\setminus \partial K))$.
We claim that $A'$ is an (induced) $\pi$-packing.  Indeed, since both $A$ and $O$ have this property, the only
way how this could be false is if there existed $P\in A\setminus K$ and $P'\in O\cap (K\setminus \partial K)$
that intersect (or contain adjacent vertices in the induced case).  However, then $PP'\in E(H)$, which
is a contradiction since $P'\in K\setminus \partial K$ and $P\not\in K$.
Observe also that $|A'|\le |A|$, as otherwise the $c$-local search algorithm
would extend $A\setminus K$ to a solution larger than $A$ (if $|A'|>|A|$, then $K\not\subseteq A$,
and thus $A\setminus K$ is obtained from $A$ by deleting less than $c$ elements).

We conclude that
$$|A\cap K|\ge |O\cap (K\setminus\partial K)|$$
for every $K\in \KK$.  Therefore, denoting by $s$ the excess of $\KK$, we have
\begin{align*}
|A|&\ge \sum_{K\in\KK} |A\cap (K\setminus \partial K)|\ge -s+\sum_{K\in\KK} |A\cap K|\\
&\ge -s+\sum_{K\in\KK} |O\cap (K\setminus \partial K)|\ge -2s+\sum_{K\in\KK} |O\cap K|\\
&\ge -2s+|O|\ge (1-\varepsilon)|O|,
\end{align*}
as required.

\end{proof}

\section*{Acknowledgements}

I would like to thank Abhiruk Lahiri for discussions about these problems.

\bibliographystyle{acm}
\bibliography{../data.bib}

\begin{thebibliography}{1}

\bibitem{dvorak2013testing}
{\sc Dvo{\v{r}}{\'a}k, Z., Kr{\'a}l', D., and Thomas, R.}
\newblock Testing first-order properties for subclasses of sparse graphs.
\newblock {\em Journal of the ACM (JACM) 60}, 5 (2013), 36.

\bibitem{dvorak2016strongly}
{\sc Dvo{\v{r}}{\'a}k, Z., and Norin, S.}
\newblock Strongly sublinear separators and polynomial expansion.
\newblock {\em SIAM Journal on Discrete Mathematics 30\/} (2016), 1095--1101.

\bibitem{har2016notes}
{\sc Har-Peled, S., and Quanrud, K.}
\newblock Notes on approximation algorithms for polynomial-expansion and
  low-density graphs.
\newblock {\em arXiv 1603.03098\/} (2016).

\bibitem{har2015approximationjrnl}
{\sc Har-Peled, S., and Quanrud, K.}
\newblock Approximation algorithms for polynomial-expansion and low-density
  graphs.
\newblock {\em SIAM Journal on Computing 46}, 6 (2017), 1712--1744.

\bibitem{zhu2009colouring}
{\sc Zhu, X.}
\newblock Colouring graphs with bounded generalized colouring number.
\newblock {\em Discrete Math. 309}, 18 (2009), 5562--5568.

\end{thebibliography}

\end{document}